\documentclass[11pt,letterpaper]{article}
\usepackage[margin=1in]{geometry}
\usepackage{lmodern}
\usepackage{amsmath}
\usepackage{amssymb}
\usepackage{amsthm}
\usepackage{amsfonts} 
\usepackage[pdftex]{graphicx} 
\usepackage[english]{babel} 
\usepackage[pdftex,linkcolor=black,pdfborder={0 0 0}]{hyperref} 
\usepackage{calc} 
\usepackage{enumerate}
\usepackage{physics} 
\usepackage{graphics}
\usepackage{tikz}

\usepackage{dsfont}
\usetikzlibrary{positioning}

 \usepackage[ruled,vlined,linesnumbered]{algorithm2e}
 \usepackage[noend]{algpseudocode}
  \usepackage{hyperref}
 \usepackage[capitalise]{cleveref}

\theoremstyle{plain}
\newtheorem{theorem}{Theorem}

\newtheorem{conjecture}[theorem]{Conjecture}

\newtheorem{lemma}[theorem]{Lemma}

\theoremstyle{definition}
\newtheorem{definition}[theorem]{Definition}

\theoremstyle{remark}
\newtheorem*{remark}{Remark}

\newcommand{\OPT}{\mathrm{OPT}}
\newcommand{\opt}{\mathrm{opt}}
\newcommand{\solM}{\mathrm{sol}(M)}
\newcommand{\tsp}{\mathrm{TSP}}
\newcommand{\eps}{\varepsilon}

\def\norm#1{\left\lVert  #1 \right\lVert}

\usepackage{listings}

\renewcommand{\P}{\mathbb{P}}

\newcommand{\rad}{\mathrm{rad}}

\title{Euclidean Capacitated Vehicle Routing in Random Setting: \newline A $1.55$-Approximation Algorithm}

\author{
Zipei Nie\footnote{Lagrange Mathematics and Computing Research Center, Huawei; Institut des Hautes Études Scientifiques, Paris, France, email: \texttt{niezipei@huawei.com}.}
\and
Hang Zhou\footnote{École Polytechnique, Institut Polytechnique de Paris, France, email: \texttt{hzhou@lix.polytechnique.fr}.}
}
\date{\today}
\begin{document}

\maketitle

\begin{abstract}
We study the unit-demand capacitated vehicle routing problem in the random setting of the Euclidean plane. The objective is to visit $n$ random terminals in a square using a set of tours of minimum total length, such that each tour visits the depot and at most $k$ terminals.

We design an elegant algorithm combining the classical sweep heuristic and Arora's framework for the Euclidean traveling salesman problem [Journal of the ACM 1998]. 
We show that our algorithm is a polynomial-time approximation of ratio at most $1.55$ asymptotically almost surely.
This improves on previous approximation ratios of $1.995$ due to Bompadre, Dror, and Orlin [Journal of Applied Probability 2007] and $1.915$ due to Mathieu and Zhou [Random Structures and Algorithms 2022]. 
In addition, we conjecture that, for any $\varepsilon>0$, our algorithm is a $(1+\varepsilon)$-approximation asymptotically almost surely.
\end{abstract}

\section{Introduction}

In the \emph{unit-demand capacitated vehicle routing problem (CVRP)}, we are given a set $V$ of $n$ \emph{terminals} and a \emph{depot} $O$.
The terminals and the depot are located in some metric space.
There is an unlimited number of identical vehicles, each of an integer \emph{capacity} $k$.
The tour of a vehicle starts at the depot and returns there after visiting at most $k$ terminals.
The objective is to visit every terminal, using a set of tours of minimum total length.
Unless explicitly mentioned, for all CVRP instances in this paper, each terminal is assumed to have \emph{unit demand}.
Vehicle routing is a basic type of problems in operations research, and several books~(see~\cite{anbuudayasankar2016models,crainic2012fleet,golden2008vehicle,toth2002vehicle} among others) have been written on those problems.

We study the \emph{Euclidean} version of the CVRP, in which all locations (the terminals and the depot) lie in the two-dimensional plane, and the distances are given by the Euclidean metric.
The Euclidean CVRP is a generalization of the Euclidean traveling salesman problem and is known to be NP-hard for all $k\geq 3$ (see~\cite{asano1997covering}).
Surprisingly, as stated in a survey of Arora~\cite{arora2003approximation}, the Euclidean CVRP is the first problem resisting Arora's famous framework on approximation schemes~\cite{arora1998polynomial}.
Whether there is a polynomial-time $(1+\eps)$-approximation for the Euclidean CVRP for any $\eps>0$ is a fundamental question and remains open regardless of numerous efforts for several decades, e.g.,~\cite{adamaszek2010ptas,arora2003approximation,asano1997covering,das2015quasipolynomial,GMZ2023,haimovich1985bounds,jayaprakash2023approximation,khachay2016ptas,MZ23_PTAS}.

Given the difficult challenges in the Euclidean CVRP, researchers turned to an analysis beyond worst case, by making some probabilistic assumptions on the distribution of the input instance. 
In 1985, Haimovich and Rinnooy~Kan~\cite{haimovich1985bounds} first studied this problem in the \emph{random setting}, where the terminals are $n$ \emph{independent, identically distributed (i.i.d.)} uniform random points in $[0,1]^2$.
An event $\mathcal{E}$ occurs \emph{asymptotically almost surely (a.a.s.)} if $\lim_{n\to\infty}\P[\mathcal{E}]=1$.
It is a long-standing open question whether, in the random setting, there is a polynomial-time $(1+\eps)$-approximation for the Euclidean CVRP a.a.s.\ for any $\eps>0$.
Haimovich and Rinnooy~Kan~\cite{haimovich1985bounds} introduced the classical \emph{iterated tour partitioning (ITP)} algorithm, and they raised the question whether, in the random setting, ITP is a $(1+\eps)$-approximation a.a.s.\ for any $\eps>0$.
Bompadre, Dror, and Orlin~\cite{bompadre2007probabilistic} showed that, in the random setting, the approximation ratio of ITP is at most 1.995 a.a.s.
Recently, Mathieu and Zhou~\cite{MZ22} showed that, in the random setting, the approximation ratio of ITP is at most 1.915 and at least $1+c_0$, for some constant $c_0>0$, a.a.s.
Consequently, designing a $(1+\eps)$-approximation in the random setting for any $\eps>0$ would require new algorithmic insights.

\paragraph{Theoretical and Practical Perspectives.}

Almost all CVRP algorithms that have theoretical guarantees in the Euclidean plane are (either completely or partially) based on ITP, e.g.~\cite{adamaszek2010ptas,altinkemer1987heuristics,altinkemer1990heuristics,asano1997covering,blauth2023improving,bompadre2006improved,bompadre2007probabilistic,das2015quasipolynomial,friggstad2021improved,GMZ2023,li1990worst,MZ22}.
However, ITP is seldom used in industry, because its practical performance is not as good as many other heuristics~\cite{viger}.

In this paper, we design an elegant algorithm (\cref{alg:main}) that is completely different from ITP. 
Our algorithm is inspired by the classical \emph{sweep heuristic} (see \cref{sec:related}), one of the most popular heuristics in practice.
Interestingly, we show that, in the random setting, our algorithm achieves the best-to-date theoretical performance, and leads to a significant progress towards a $(1+\eps)$-approximation for any $\eps>0$.
See \cref{sec:results}.

\subsection{Our Results}
\label{sec:results}

\paragraph{Algorithm.}

We present a simple polynomial-time algorithm for the Euclidean CVRP. See \cref{alg:main}.
For each terminal $v$, let $\theta(v)\in [0,2\pi)$ denote the polar angle of $v$ with respect to $O$.
First, we sort all terminals in nondecreasing order of $\theta(v)$.
Let $M\geq 1$ be a constant integer parameter. 
Then we decompose the sorted sequence into subsequences, each consisting of $M k$ consecutive terminals, except possibly for the last subsequence containing less terminals.
Next, for the terminals in each subsequence, we compute a near-optimal solution to the CVRP using Arora's framework for the Euclidean traveling salesman problem~\cite{arora1998polynomial}, see also \cref{sec:preli}. 

\begin{algorithm}[h]
    \caption{Algorithm for the CVRP in $\mathbb{R}^2$. Constant integer parameter $M\ge 1$.}
    \label{alg:main}
    \KwIn{set $V$ of $n$ terminals in $\mathbb{R}^2$, depot $O\in \mathbb{R}^2$, capacity $k\in\{1,2,\dots,n\}$}
    \KwOut{set of tours covering all terminals in $V$}
    Sort the terminals in $V$ into $u_1,u_2,\dots,u_n$ such that $\theta(u_1)\le \theta(u_2)\le \cdots \le \theta(u_n)$
    
    \For{$i\gets 1$ \KwTo $\left\lceil\frac{ n  }{Mk}\right\rceil$}{
    $V_i\gets \left\{u_j: (i-1)\cdot Mk<j\le i\cdot M k\right\}$
    
    Compute a $\left(1+\frac{1}{M}\right)$-approximate solution $S_i$ to the subproblem $(V_i,O,k)$
    \Comment{\cref{lem:Arora-modified}}
    }
    \Return union of all $S_i$
\end{algorithm}

Our main result shows that, in the random setting, \cref{alg:main} has an approximation ratio at most $1.55$ a.a.s., see \cref{thm:main}.
This improves on previous ratios of $1.995$ due to Bompadre, Dror, and Orlin~\cite{bompadre2007probabilistic} and $1.915$ due to Mathieu and Zhou~\cite{MZ22}. 
Furthermore, we conjecture that, in the random setting, \cref{alg:main} is a $(1+\eps)$-approximation for any $\eps>0$ a.a.s., see \cref{conjecture}.

\begin{theorem}
\label{thm:main}
Consider the unit-demand Euclidean CVRP with a set~$V$ of $n$ terminals {that} are i.i.d.\ uniform random points in {$[0,1]^2$}, a fixed depot $O\in \mathbb{R}^2$, and a capacity~$k$ that takes an arbitrary value in $\{1,2,\dots,n\}$. 
For any constant integer $M\geq 10^5$, \cref{alg:main} with parameter $M$ is a polynomial-time approximation of ratio at most $1.55$ asymptotically almost surely.
\end{theorem}

\begin{conjecture}
\label{conjecture}
Consider the unit-demand Euclidean CVRP with $V$, $O$, and $k$ defined in \cref{thm:main}. 
For any $\eps>0$, there exists a positive constant integer $M$ depending on $\eps$, such that \cref{alg:main} with parameter $M$ is a polynomial-time $(1+\eps)$-approximation asymptotically almost surely.
\end{conjecture}

\subsection{Overview of Techniques}

A main contribution in our analysis is the novel concepts of \emph{$R$-radial cost} and \emph{$R$-local cost}. These are generalizations of the classical \emph{radial cost} and \emph{local cost} introduced by Haimovich and Rinnooy~Kan~\cite{haimovich1985bounds}.

\begin{definition}[$R$-radial cost]
\label{def:R-radial}
For any $R\in\mathbb{R}^+\cup \left\{0,\infty\right\}$, define the \emph{$R$-radial cost} $\rad_R$ by
    \[\rad_R :=\frac{2}{k} \sum_{v\in V}\min\left\{ d(O,v) , R\right\}.\]
\end{definition}

\begin{definition}[$R$-local cost]\label{def:R-local}
For any $R\in\mathbb{R}^+\cup \left\{0,\infty\right\}$, define the \emph{$R$-local cost} $T^*_R$ as the minimum cost of a traveling salesman tour on the set of points $\left\{v\in V:d(O,v)\ge R\right\}$.
\end{definition}

Using the $R$-radial cost and the $R$-local cost, we establish a new lower bound (\cref{thm:lower}) on the cost of an optimal solution.
This lower bound is a main novelty of the paper.
It unites both classical lower bounds from \cite{haimovich1985bounds}: when $R=0$, it leads asymptotically to one classical lower bound, which is the local cost; and when $R=\infty$, it leads asymptotically to the other classical lower bound, which is the radial cost.
The proof of \cref{thm:lower} {is} in \cref{sec:lower}.

\begin{theorem}\label{thm:lower}
Consider the unit-demand Euclidean CVRP with any set $V$ of $n$ terminals in $\mathbb{R}^2$, any depot $O\in\mathbb{R}^2$, and any capacity $k\in \mathbb{N}^+$.
Let $\opt$ denote the cost of an optimal solution. For any $R\in\mathbb{R}^+\cup \left\{0,\infty\right\}$, we have
\[\opt\ge T^*_R+ \rad_R -{\frac{3\pi D}{2}},\]
where $D$ denotes the diameter of $V\cup\{O\}$. 
\end{theorem}

Next, we establish an upper bound (\cref{thm:upper}) on the cost of the solution in \cref{alg:main} using the $0$-local cost and the $\infty$-radial cost.
The proof of Theorem \ref{thm:upper} is in \cref{sec:upper}.

\begin{theorem}\label{thm:upper}
Consider the unit-demand Euclidean CVRP with any set $V$ of $n$ terminals in $\mathbb{R}^2$, any depot $O\in\mathbb{R}^2$, and any capacity $k\in \mathbb{N}^+$.
For any positive integer $M$, let $\solM$ denote the cost of the solution returned by \cref{alg:main} with parameter $M$.
Then we have
\[\solM \le \left(1+\frac{1}{M}\right)\left(T_0^* +\rad_\infty+ \frac{3 \pi D}{2} \left\lceil\frac{ n }{Mk}\right\rceil \right),\]
where $D$ denotes the diameter of $V\cup\{O\}$. 
\end{theorem}

Note that both \cref{thm:lower} and \cref{thm:upper} hold for any set of terminals, not only in the random setting, and can be of independent interest.

In the random setting, in order to compute the approximation ratio of \cref{alg:main}, we set $R$ to be some well-chosen value so that the ratio between the upper bound (\cref{thm:upper}) and the lower bound (\cref{thm:lower}) is {small}. The proof of \cref{thm:main} is in \cref{sec:proof-main} and relies on two technical inequalities proved in \cref{appen}.

\begin{remark}
In the random setting, the term containing $D$ in the bound in \cref{thm:lower} (resp.\ \cref{thm:upper}) becomes negligible. 
The upper bound in \cref{thm:upper} is then
asymptotically identical to the upper bound for the ITP algorithm established in~\cite{altinkemer1990heuristics}. 
As an immediate consequence, one can prove that the approximation ratio of ITP is at most $1.55$ using a similar analysis as in this paper.
Note that the upper bound for ITP established in \cite{altinkemer1990heuristics} is tight~\cite[Lemma 7]{MZ22}, and the tightness is exploited to show that ITP is at best a $(1+c_0)$-approximation for some constant $c_0>0$~\cite{MZ22}.
However, we are not aware of any instance on which the upper bound in \cref{thm:upper} is tight for \cref{alg:main}, and we believe that the performance of \cref{alg:main} is actually better than the announced upper bound.
\end{remark}

\subsection{Related Work}
\label{sec:related}
\paragraph{Sweep Heuristic.}
The classical \emph{sweep heuristic} is well-known and commercially available for vehicle routing problems in the plane. 
At the beginning, all terminals are sorted according to their polar angles with respect to the depot.
For each $k$ consecutive terminals in the sorted sequence, a tour is obtained by computing a traveling salesman tour (exactly or approximately) on those terminals. 
Some implementations include a post-optimization phase in which vertices in adjacent tours may be exchanged to reduce the overall cost. 
The first mentions of this type of method are found in a book by Wren~\cite{wren1971computers} and in a paper by Wren and Holliday~\cite{wren1972computer}, while the sweep heuristic is commonly attributed to Gillett and Miller~\cite{gillett1974heuristic} who popularized it. See also surveys~\cite{cordeau2007vehicle,laporte1992vehicle,laporte2000classical} and the book~\cite{toth2002vehicle}.

Our algorithm (\cref{alg:main}) is an adaptation of the sweep heuristic: instead of forming groups each of $k$ consecutive terminals, we form groups each of $Mk$ consecutive terminals, for some positive constant integer $M$.
Then for each group, we compute a solution consisting of a constant number of tours using Arora's framework~\cite{arora1998polynomial}.
Note that it is possible to replace Arora’s framework by a heuristic in order to make our algorithm more practical.

Our algorithm is simple, so it can be easily adapted to other vehicle routing problems, e.g.,\ distance-constrained vehicle routing~\cite{DMZ23,friggstad2014approximation,li1992distance,nagarajan2012approximation}.

\paragraph{Euclidean CVRP.}
Despite the difficulty of the Euclidean CVRP, there has been progress on several special cases in the deterministic setting. A series of papers designed \emph{polynomial-time approximation schemes (PTAS's)} for small $k$:
Haimovich and Rinnooy Kan~\cite{haimovich1985bounds} gave a PTAS when $k$ is constant;
Asano et al.~\cite{asano1997covering}  extended the techniques in~\cite{haimovich1985bounds} to achieve a PTAS for $k=O(\log n/\log\log n)$;
and Adamaszek, Czumaj, and Lingas~\cite{adamaszek2010ptas} designed a PTAS for $k\leq 2^{\log^{f(\eps)}(n)}$.
For higher dimensional Euclidean metrics, Khachay and Dubinin~\cite{khachay2016ptas} gave a PTAS for fixed dimension $\ell$ and $k=O(\log^{\frac{1}{\ell}}(n))$.
For arbitrary $k$ and the two-dimensional plane, Das and Mathieu~\cite{das2015quasipolynomial} designed a quasi-polynomial time approximation scheme, whose running time was recently improved  to $n^{O(\log^6 (n)/\eps^5)}$ by Jayaprakash and Salavatipour~\cite{jayaprakash2023approximation}.

\paragraph{Probabilistic Analyses.} The random setting in which the terminals are i.i.d.\ uniform random points is perhaps the most natural probabilistic setting.
The Euclidean CVRP in the random setting has been studied in several special cases. 
In one special case when the capacity is infinite, Karp~\cite{karp1977probabilistic} gave a polynomial-time $(1+\eps)$-approximation a.a.s.\ for any $\eps>0$.
In another special case when $k$ is fixed, Rhee~\cite{rhee1994probabilistic} and Daganzo~\cite{daganzo1984distance} analyzed the value of an optimal solution.

\paragraph{CVRP in Other Metrics.}
The CVRP has been extensively studied on general metrics~\cite{altinkemer1990heuristics,blauth2023improving,bompadre2006improved,haimovich1985bounds,li1990worst}, trees and bounded treewidth~\cite{asano2001new,becker2018tight,becker2019framework,jayaprakash2023approximation,MZ23_PTAS}, planar and bounded-genus graphs~\cite{becker2017quasi,becker2019ptas,cohen2020light}, graphic metrics~\cite{MZ23graphic},
graphs of bounded highway dimension~\cite{becker2018polynomial}, and minor-free graphs~\cite{cohen2020light}.

\paragraph{CVRP with Arbitrary Demands.}
A natural way to generalize the unit demand version of the CVRP is to allow terminals to have arbitrary unsplittable demands, which is called the \emph{unsplittable} version of the CVRP.
On general metrics, Altinkemer and  Gavish~\cite{altinkemer1987heuristics} first studied the approximation of this problem. 
Recently, the approximation ratio was improved by Blauth, Traub, and Vygen~\cite{blauth2023improving}, and further by Friggstad, Mousavi, Rahgoshay, and Salavatipour~\cite{friggstad2021improved}.
This problem has also been studied in the Euclidean plane by Grandoni, Mathieu, and Zhou~\cite{GMZ2023} and on trees by Mathieu and Zhou~\cite{MZ22_unsplittable}.

\subsection{Notations and Preliminaries}
\label{sec:preli}
For any two points $u$ and $v$ in $\mathbb{R}^2$, let $d(u,v)$ denote the \emph{distance} between $u$ and $v$ in $\mathbb{R}^2$. For any curve $s$ in $\mathbb{R}^2$, let $\norm{s}$ denote the length of $s$; and for any set $S$ of curves in $\mathbb{R}^2$, let $\norm{S}:=\sum_{s\in S}\norm{s}$. 
For any set $U$ of points in $\mathbb{R}^2$, the \emph{convex hull} of $U$ is the minimal convex set in $\mathbb{R}^2$ containing $U$.

\paragraph{Arora's Framework.}
The following lemma gives a PTAS for the Euclidean CVRP when the capacity $k$ is at least a constant fraction of the number of terminals. 
Its proof is a straightforward adaptation of Arora's framework~\cite{arora1998polynomial} for the Euclidean TSP.

\begin{lemma}[adaptation of \cite{arora1998polynomial}]
\label{lem:Arora-modified}
Let $M\geq 1$ be an integer constant. 
Consider the unit-demand Euclidean capacitated vehicle routing problem with any finite set $U$ of terminals in $\mathbb{R}^2$, any depot $O\in\mathbb{R}^2$, and any capacity $k$ such that $|U|\leq Mk$.
Then there exists a polynomial-time $\left(1+\frac{1}{M}\right)$-approximation algorithm.
\end{lemma}

\begin{proof}
Recall that Arora's algorithm defines a randomized hierarchical quadtree decomposition, such that a near-optimal solution intersects the boundary of each square only $O_M(1)$ times and those crossings happen at one of a small set of prespecified points, called \emph{portals}, and then uses a polynomial time dynamic program to find the best solution with this structure.

In \cite{asano1997covering} it was observed that, when the number of tours in an optimal solution is $O_M(1)$, there is a near-optimal solution in which the overall number of subtours passing through each square (via portals) is $O_M(1)$. Furthermore, one can guess the number of terminals covered by each such subtour within a polynomial number of options. This leads to a polynomial number of configurations of subtours inside each square, which ensures the polynomial running time of a natural dynamic program.
\end{proof}

\section{Proof of Theorem~\ref{thm:lower}}\label{sec:lower}
Let $\OPT$ denote an optimal solution to the CVRP. 
Let $C$ denote the circle centered at $O$ with radius $R$. Suppose that the union of the tours in $\OPT$ intersects $C$ at $2 t$ points, denoted by $y_1, y_2,\ldots, y_{2t}$ in clockwise order. For notational convenience, we let $y_{2t+1}:=y_1$. Let $D'$ denote the diameter\footnote{We adopt the convention that the diameter of an empty set is zero.} of $\{y_i: 1\le i\le 2t\}$. Let $C_1,C_2,\ldots,C_t$ be $t$ continuous curves that correspond to the intersection between $\OPT$ and the closure of the exterior of $C$. 

\begin{lemma}
\label{lem:sum-Ci}
We have
\[\sum_{i=1}^t \norm{C_i}  \ge T_R^*-\frac{3\pi D'}{2}.\]
\end{lemma}

\begin{proof}
Let $Z$ denote the set of segments $y_i y_{i+1}$ for all $1\leq i\leq 2t$.
Let $Z_{\rm{odd}}$ (resp.\ $Z_{\rm{even}}$) denote the set of segments $y_i y_{i+1}$ for all $1\leq i\leq 2t$ such that $i$ is odd (resp.\ even).
Let $Z^*$ be one of $Z_{\rm{odd}}$ and $Z_{\rm{even}}$ that has a smaller total length, breaking ties arbitrarily.
Let $W$ denote the union of the curves $C_1,C_2,\ldots,C_t$, the segments in $Z$, and the segments in $Z^*$. 
Then $W$ is a connected graph with no odd degree vertices. So $W$ has an Eulerian path.
Since $W$ visits all vertices $v\in V$ such that $d(O,v)\geq R$, the total length of $W$ is at least $T_R^*$ by \cref{def:R-local}. 
Hence 
\[\norm{Z}+\norm{Z^*}+\sum_{i=1}^t \norm{C_i} \geq T^*_R.\]

We note that $\norm{Z}$ equals to the perimeter of the convex hull of $\{y_i :1\le i\le 2t\}$, which is at most $\pi D'$ by \cite{szasz1917extremaleigenschaft}. Since $\norm{Z^*}\leq \frac{1}{2}\norm{Z}$, we have \[\norm{Z}+\norm{Z^*}\leq \frac{3\norm{Z}}{2} \le \frac{3\pi D'}{2}.\] The claim follows.
\end{proof}

The following lemma introduces a key new idea in our paper.

\begin{lemma} 
\label{lem:normS}
Let $s$ be any tour in $\OPT$.
Let $V_s \subseteq V$ denote the set of points in $V$ that are visited by~$s$.
Let $U_s\subseteq \left\{1,2,\ldots, t\right\}$ denote the set of indices $i$ such that $C_i$ is part of $s$. 
We have
\begin{equation}
    \label{eqn:s}
    \norm{s}\ge\sum_{i\in U_s}\norm{C_i}+\frac{2}{k}\sum_{v\in V_s}\min\left\{d(O,v),R\right\}.
\end{equation}
\end{lemma}
\begin{proof}
There are two scenarios to consider. 

First, if $U_s$ is empty, then we have 
\[\norm{s}\ge 2\max_{v\in V_s}d(O,v)\ge \frac{2}{|V_s|}\sum_{v\in V_s}\min\left\{d(O,v),R\right\}.\]
The claim follows since $|V_s|\leq k$.

Second, if $U_s$ is nonempty, then the tour $s$ must first travel through a path to a point on $C$, paying at least $R$, then visit all curves $C_i$ for $i\in U_s$, and finally, travel from a point on $C$ back to the depot, paying at least $R$. Thus we have
\[\norm{s}\ge 2R+\sum_{i\in U_s}\norm{C_i}.\]
The claim follows since $\displaystyle 2R\geq \frac{2}{|V_s|}\sum_{v\in V_s}\min\left\{d(O,v),R\right\}$ and  $|V_s|\leq k$.
\end{proof}

Summing \eqref{eqn:s} over all tours $s\in \OPT$, we have
\begin{align*}
    \opt =&\sum_{s\in \OPT} \norm{s}\\
    \ge &\sum_{s\in \OPT}\sum_{i\in U_s}\norm{C_i}+\frac{2}{k}\sum_{s\in \OPT}\sum_{v\in V_s}\min\left\{d(O,v),R\right\}\\
    = & \sum_{i=1}^t \norm{C_i}
    +\frac{2}{k} \sum_{v\in V}\min\left\{d(O,v),R\right\}\\
    \ge& T_R^* - \frac{3\pi D'}{2} +\rad_R,
\end{align*}
where the last inequality follows from \cref{{lem:sum-Ci}} and the definition of $R$-radial cost (\cref{def:R-radial}). Since each $y_i$ lies in the convex hull of $V\cup O$, we have $D'\le D$, so the claim in \cref{thm:lower} follows.
\begin{remark} One can show that $D'$ is at most the diameter of $V$. This is because each $y_i$ is contained in the projection of the convex hull of $V$ onto the disk enclosed by $C$, and the projection onto any closed convex set is $1$-Lipschitz\footnote{We say that a function $f$ is $1$-Lipschitz if $d(f(x),f(y))\le d(x,y)$ for all $x$ and $y$.}; see \cite[Theorem 1.2.4]{schneider2014convex} for example. 
\end{remark}

\section{Proof of Theorem~\ref{thm:upper}}\label{sec:upper}
Let $i$ be any integer with $1\le i\le \left\lceil\frac{n}{Mk}\right\rceil$. 
Let the point set $V_i$ and the solution $S_i$ be defined in Algorithm \ref{alg:main}. Let $S_i^*$ denote an optimal solution to the subproblem $(V_i,O,k)$.
Since $S_i$ is a $\left(1+\frac{1}{M}\right)$-approximate solution, we have $\norm{S_i}\leq \left(1+\frac{1}{M}\right)\cdot\norm{S_i^*}$.
Let $\tsp_i$ denote the minimum cost of a traveling salesman tour on the set of points $V_i\cup\{O\}$.
By \cite[Lemma 2]{altinkemer1990heuristics}, we have
\[
\norm{S_i^*}\leq \tsp_i+ \frac{2}{k}\sum_{v\in V_i} d(O,v).
\]
Thus
\begin{equation}
\label{eqn:S_i}
\norm{S_i}\leq \left(1+\frac{1}{M}\right)\left(\tsp_i+ \frac{2}{k}\sum_{v\in V_i} d(O,v)\right).    
\end{equation}

Let $t^*$ be an optimal traveling salesman tour on the set of points $V$. If the polar angles of points in $V_i$ have a span of at most $\pi$, let $Y_i$ be the interior of the convex hull of $V_i\cup\{O\}$; otherwise, let $Y_i$ be the exterior of the convex hull of $(V\setminus V_i)\cup\{O\}$. By \cite[Theorem 3]{karp1977probabilistic}\footnote{Although \cite[Theorem 3]{karp1977probabilistic} assumes that $Y$ is a rectangle, the arguments extend trivially to any polygon or the exterior of any polygon.}, we have 
\[\tsp_i-\norm{t^*\cap Y_i}\le \frac{3}{2} \;\mathrm{per}(Y_i),\] 
where $\mathrm{per}(Y_i)$ denotes the perimeter of $Y_i$. 
Since either $Y_i$ or the complement of $Y_i$ is convex with diameter at most $D$, the perimeter of $Y_i$ is at most $\pi D$ by \cite{szasz1917extremaleigenschaft}.
Thus 
\begin{equation}
\label{eqn:T_i^*}
    \tsp_i\le \norm{t^*\cap Y_i} +\frac{3 \pi D}{2}.
\end{equation}

In order to bound $\sum_i\norm{t^*\cap Y_i}$, we need the following lemma.

\begin{lemma}\label{lem:no-intersection}
For any $i$ and $j$ with $1\le i<j\le \left\lceil\frac{n}{Mk}\right\rceil$,  $Y_i$ and $Y_j$ do not intersect. 
\end{lemma}
\begin{proof}
For each $v\in V$, let $\theta(v)\in [0,2\pi)$ denote the polar angle of $v$ respect to $O$. By the definition of $V_i$ and $V_j$, we have \begin{equation}\label{eq:seperation}
    0\le \max_{v\in V_i}\theta(v)\le \max_{v\in V_j}\theta(v)<2\pi.
\end{equation} Hence either 
\begin{equation}\label{eq:Y1}
    \max_{v\in V_i} \theta(v)-\min_{v\in V_i} \theta(v)\le \pi
\end{equation}
or 
\begin{equation}\label{eq:Y2}
    \max_{v\in V_j} \theta(v)-\min_{v\in V_j} \theta(v)\le \pi.
\end{equation}

If only (\ref{eq:Y1}) holds, then by definition, $Y_i$ is the interior of the convex hull of $V_i\cup\{O\}$, which is contained in the exterior of the convex hull of $(V\setminus V_j)\cup\{O\}$. Thus $Y_i$ and $Y_j$ do not intersect.

If only (\ref{eq:Y2}) holds, then $Y_i$ and $Y_j$ do not intersect for the same argument.

Suppose that both (\ref{eq:Y1}) and (\ref{eq:Y2}) hold. Let $Z_i$ be the set \[Z_i:= \left\{x\in\mathbb{R}^2 : \max_{v\in V_i} \theta(v)< \theta(x)< \min_{v\in V_i} \theta(v)\right\},\]
and $Z_j$ be the set \[Z_j:= \left\{x\in\mathbb{R}^2 : \max_{v\in V_j} \theta(v)< \theta(x)< \min_{v\in V_j} \theta(v)\right\}.\]
Then by (\ref{eq:Y1}) and (\ref{eq:Y2}), $Z_i$ and $Z_j$ are convex sets. By the definition of $Y_i$ and $Y_j$, we have $Y_i\subset Z_i$ and $Y_j\subset Z_j$.
Since $Z_i$ and $Z_j$ do not intersect, $Y_i$ and $Y_j$ do not intersect. 
\end{proof}

Therefore, we have 
\begin{align*}
    \solM=& \sum_{i=1}^{\left\lceil\frac{n}{Mk}\right\rceil} \norm{S_i}\\
    \le& \left(1+\frac{1}{M}\right)\left(\sum_{i=1}^{\left\lceil\frac{n}{Mk}\right\rceil} \tsp_i +\frac{2}{k}\sum_{v\in V}d(O,v)\right)\\
    \le&\left(1+\frac{1}{M}\right)\left(\sum_{i=1}^{\left\lceil\frac{n}{Mk}\right\rceil} \norm{t^*\cap Y_i}+\frac{3 \pi D}{2} \left\lceil\frac{n }{Mk}\right\rceil +\rad_\infty\right),
\end{align*}
where the first inequality follows from \eqref{eqn:S_i} and the fact that $\bigcup_i V_i=V$, and the last inequality follows from \eqref{eqn:T_i^*} and the definition of $\infty$-radial cost (\cref{def:R-radial}).
Using \cref{lem:no-intersection} and the definition of $0$-local cost (\cref{def:R-local}), we have
\[\sum_{i=1}^{\left\lceil\frac{n}{Mk}\right\rceil} \norm{t^*\cap Y_i}\leq \norm{t^*}=T_0^*.\]
The claim follows. 

\section{Proof of Theorem~\ref{thm:main}}
\label{sec:proof-main}
In this section, we prove a strong law for the approximation ratio of \cref{alg:main}, as presented in \cref{strong-form}. The setting is similar to those in the main results of \cite{beardwood1959shortest} and \cite{haimovich1985bounds}. Since almost sure convergence implies convergence in probability, Theorem \ref{strong-form} implies Theorem \ref{thm:main}.

\begin{theorem}\label{strong-form}
Let $v_1,v_2,\ldots$ be an infinite sequence of i.i.d.\ uniform random points in $[0,1]^2$. Let $O$ be a point in $\mathbb{R}^2$. Let $k_1,k_2,\ldots$ be an infinite sequence of positive integers. Let $M\geq 10^5$ be a positive integer. 
For each positive integer $n$, consider the unit-demand Euclidean CVRP with the set of terminals $V=\{v_1,\ldots, v_n\}$, the depot $O$, and the capacity $k_n$. Let $\opt$ denote the cost of an optimal solution, and $\solM$ denote the cost of the solution returned by \cref{alg:main} with parameter $M$. Then we have
\[\limsup_{n\to\infty}\frac{\solM}{\opt}<1.55\]
almost surely.
\end{theorem}

\begin{proof}
Let $v$ denote a uniform random point in $[0,1]^2$. Throughout this section, we always take $R:=\frac{3}{4}\;\mathbb{E}\left(d(O,v)\right)$ and $D$ to be the diameter of $[0,1]^2\cup \{O\}$. Note that $R$ and $D$ are deterministic real numbers which do not depend on $n$. Let $T_0^*$ and $T_R^*$ denote the $0$-local and $R$-local costs respectively. Let $\rad_\infty$ and $\rad_R$ denote the $\infty$-radial and $R$-radial costs respectively. By \cref{thm:lower} and \cref{thm:upper}, we have
\begin{align*}
    &\limsup_{n\to\infty}\frac{\solM}{\opt}\\
    \le &\limsup_{n\to \infty}\frac{\left(1+\frac{1}{M}\right)\left(T_0^* +\rad_\infty+ \frac{3 \pi D}{2} \left(\frac{n }{Mk_n}+1\right)  \right)}{T^*_R+\rad_R-\frac{3\pi D}{2}} \\
    \le &\limsup_{n\to \infty}\max\left\{\frac{\left(1+\frac{1}{M}\right)\left(T_0^*+\frac{3 \pi D}{2}\right)}{T_R^*-\frac{3\pi D}{2}}, \frac{\left(1+\frac{1}{M}\right)\left(\rad_\infty+\frac{3\pi D n }{2Mk_n}\right)}{\rad_R}\right\}\\
    = & \left(1+\frac{1}{M}\right)\max\left\{\limsup_{n\to\infty}\frac{T_0^*+\frac{3 \pi D}{2} }{T_R^*-\frac{3\pi D}{2}}, \limsup_{n\to\infty}\frac{\rad_\infty+\frac{3\pi D n }{2Mk_n}}{\rad_R}\right\}
\end{align*} 
almost surely.

The validity of \cref{strong-form} relies on the upper bounds on both of the limit superiors. We will prove \cref{lem:TvsT} and \cref{lem:radvsrad} in \cref{subsec:first-limsup} and \cref{subsec:second-limsup} respectively.

\begin{lemma}\label{lem:TvsT}
We have
\[\limsup_{n\to\infty}\frac{T_0^*+\frac{3 \pi D}{2}}{T_R^*-\frac{3\pi D}{2}}\le  \frac{48}{31}\]
almost surely. 
\end{lemma}

\begin{lemma}\label{lem:radvsrad}
We have
\[\limsup_{n\to\infty}\frac{\rad_\infty+\frac{3\pi D n }{2M k_n}}{\rad_R}
\le \frac{48}{31}\left(1+\frac{15 \pi }{4M}\right)\]
almost surely.
\end{lemma}

Assuming \cref{lem:TvsT} and \cref{lem:radvsrad}, for any $M \ge 10^5$, we have
\begin{align*}
&\limsup_{n\to\infty}\frac{\solM}{\opt}\\
\le    & \left(1+\frac{1}{M}\right)\max\left\{\limsup_{n\to\infty}\frac{T_0^*+\frac{3 \pi D}{2}}{T_R^*-\frac{3\pi D}{2}}, \limsup_{n\to\infty}\frac{\rad_\infty+\frac{3\pi D n  }{2Mk_n}}{\rad_R}\right\}\\
    \le&\frac{48}{31}\left(1+\frac{1}{M}\right)\left(1+\frac{15 \pi }{4M}\right)\\
    <&1.55
\end{align*} 
almost surely. 
\end{proof}
\subsection{Proof of \cref{lem:TvsT}}\label{subsec:first-limsup}
Let $\lambda_R$ denote the measure of the set $\left\{x\in [0,1]^2 : d(O,x)>R\right\}$.
Let $S_R(n)$ denote the size of set $\{1\le i\le n: d(O,v_i)> R\}$. By the strong law of large numbers, we have 
\[\lim_{n\to\infty} \frac{S_R(n)}{n}=\lambda_R\]
almost surely. We will prove the following bound on $\lambda_R$ in \cref{appen}.
\begin{lemma}\label{lemma-lambda}
For $R=\frac{3}{4}\;\mathbb{E}\left(d(O,v)\right)$, we have $\lambda_R \ge \frac{31}{48}$.
\end{lemma} 

By \cref{lemma-lambda}, $S_R(n)\to\infty$ as $n\to \infty$.  By applying the main result of \cite{beardwood1959shortest} to the infinite sequence $v_1,v_2,\ldots$ and its intersection with the set $\left\{x\in [0,1]^2 : d(O,x)>R\right\}$, we have \[ \lim_{n\to \infty} \frac{T_0^*}{\sqrt{n}}= \lim_{n\to \infty} \frac{T^*_R}{ \sqrt{\lambda_R S_R(n)}}>0\]
almost surely.
Thus
\begin{align*}
    &\lim_{n\to\infty}\frac{T_0^*+\frac{3 \pi D}{2}}{T_R^*-\frac{3\pi D}{2}}\\
    =&\lim_{n\to\infty}\frac{\frac{T_0^*}{\sqrt{n}}+\frac{3\pi D}{2\sqrt{n}}}{\lambda_R\sqrt{\frac{S_R(n)}{\lambda_R\, n} } \frac{T_R^*}{\sqrt{ \lambda_R S_R(n)}}-\frac{3\pi D}{2\sqrt{n}}}\\
    =&\frac{1}{\lambda_R}
\end{align*}
almost surely. By \cref{lemma-lambda}, the claim follows.

\subsection{Proof of \cref{lem:radvsrad}}\label{subsec:second-limsup}
By the strong law of large numbers, we have \[\lim_{n\to\infty}\frac{k_n\;\rad_\infty}{2n}=\mathbb{E}\left(d(O,v)\right)\] and \[\lim_{n\to\infty}\frac{k_n\;\rad_R}{2n}=\mathbb{E}\left(\min\left\{d(O,v), R\right\}\right)\] almost surely. 
Thus
\begin{align*}
&\lim_{n\to\infty}\frac{\rad_\infty+\frac{3\pi  D n }{2M k_n}}{\rad_R}\\
=&\lim_{n\to\infty}\frac{\frac{k_n\;\rad_\infty}{2n}+\frac{3\pi D}{4M}}{\frac{k_n\;\rad_R}{2n}}\\
=&\frac{\mathbb{E}\left(d(O,v)\right)+ \frac{3\pi D}{4M}}{\mathbb{E}\left(\min\left\{d(O,v), R\right\}\right)}
\end{align*}
almost surely.
We need the following inequality whose proof is in \cref{appen}.
\begin{lemma}\label{lemma-mv-2}
For $R=\frac{3}{4}\;\mathbb{E}\left(d(O,v)\right)$, we have \[\mathbb{E}\left(\min\left\{d(O,v), R\right\}\right) \ge \frac{31}{48}\;\mathbb{E}\left(d(O,v)\right).\]
\end{lemma}
From \cref{lemma-mv-2}, we obtain
\begin{equation}\label{eq:radvsrad1}
\lim_{n\to\infty}\frac{\rad_\infty+\frac{3\pi D n }{2M k_n}}{\rad_R}\le\frac{48\left(\mathbb{E}\left(d(O,v)\right)+\frac{3\pi D}{4M}\right)}{31\;\mathbb{E}\left(d(O,v)\right)}
\end{equation}
almost surely.

\begin{lemma}\label{bound-f}
We have
\[D\le 5\; \mathbb{E}\left(d(O,v)\right).\]
\end{lemma}
\begin{proof}
Let $O_c=\left(\frac{1}{2},\frac{1}{2}\right)\in\mathbb{R}^2$ denote the center of the square $[0,1]^2$. Let $\overline{v}$ denote the reflection of $v$ across the point $O_c$. Then we have 
\[\frac{d(O,v)+d(O,\overline{v})}{2}\ge\frac{d(v, \overline{v})}{2}= d(O_c,v).\]
Because $v$ and $\overline{v}$ have the same distribution, we have \[\mathbb{E}\left(d(O,v)\right) \ge \mathbb{E}\left(d(O_c,v)\right).\]
We use a closed-form formula of $\mathbb{E}\left(d(O_c,v)\right)$ established in \cref{subsec:closed-form} (\cref{closed-form-g1}) to obtain 
\[ \mathbb{E}\left(d(O_c,v)\right)= \frac{\sqrt{2}+\log \left(1+\sqrt{2}\right)}{6} \ge \frac{\sqrt{2}}{4}.\]
Therefore, by the definition of $D$, we have \[D\le \sqrt{2}+\mathbb{E}\left\{d(O,v)\right\}\le 4\;\mathbb{E}\left(d(O_c,v)\right)+\mathbb{E}\left\{d(O,v)\right\}\le 5\; \mathbb{E}\left(d(O,v)\right).\]  
\end{proof}

By \cref{eq:radvsrad1} and \cref{bound-f}, we proved \cref{lem:radvsrad}.

\appendix
\begin{appendix}
\section{Proofs of \cref{lemma-lambda} and \cref{lemma-mv-2}}\label{appen}
In this section, we prove two inequalities involving $\mathbb{E}\left(d(O,v)\right)$, $\mathbb{E}\left(\min\{d(O,v),R\}\right)$, and $\lambda_R$. Here $v$ is a uniform random point in $[0,1]^2$, $R$ is $\frac{3}{4} \; \mathbb{E}\left(d(O,v)\right)$, and $\lambda_R$ is the measure of the set $\{x\in [0,1]^2: d(O,x)>R\}$. For convenience, we rewrite these terms as explicit functions of $O$, that is, we let 
\begin{align*}
g_1(O)&:=\mathbb{E}\left(d(O,v)\right),\\
g_2(O)&:=\mathbb{E}\left(\min\{d(O,v),R\}\right),\\
g_3(O)&:=\lambda_R.
\end{align*}

The proofs of \cref{lemma-lambda} and \cref{lemma-mv-2} consist of several steps. 
In \cref{sec:Lipschitz}, we establish the Lipschitz conditions for $g_1$, $g_2$, and $g_3$. 
In \cref{sec:epsilon}, first, we prove the claims when $O$ is located far away from the unit square; next, we build a fixed-sized $\eps$-net $N$ of a bounded region, and use validated numerics to establish rigorous inequalities for all points in $N$; and finally, we prove the claims in \cref{lemma-lambda} and \cref{lemma-mv-2}
for any point $O$ in $\mathbb{R}^2$.
In \cref{subsec:closed-form}, we prove the closed-form formulas that are used in \cref{sec:epsilon}.

\subsection{Lipschitz Continuity of $g_1$, $g_2$, and $g_3$}
\label{sec:Lipschitz}
In this subsection, we compute the Lipschitz constants of the functions $g_1$, $g_2$, and $g_3$.

\begin{lemma}\label{lem:Lipschitz-g1}
For any $O,O'\in \mathbb{R}^2$, we have \[\left|g_1(O)-g_1(O')\right|\le d(O,O').\]
\end{lemma}
\begin{proof}
We have
\begin{align*}
    &\left|g_1(O)-g_1(O')\right|\\
    =&\left|\mathbb{E}\left(d(O,v)-d(O',v)\right)\right|\\
    \le&\mathbb{E}\left(\left|d(O,v)-d(O',v)\right|\right)\\
    \le&\mathbb{E}\left(d(O,O')\right)\\
    =&d(O,O').
\end{align*}
\end{proof}

\begin{lemma}\label{lem:Lipschitz-g2}
For any $O,O'\in \mathbb{R}^2$, we have \[\left|g_2(O)-g_2(O')\right|\le d(O,O').\]
\end{lemma}
\begin{proof}
By \cref{lem:Lipschitz-g1}, we have
\begin{align*}
    &\left|g_2(O)-g_2(O')\right|\\
    =&\left|\mathbb{E}\left(\min\left\{d(O,v),\frac{3}{4} g_1(O)\right\}-\min\left\{d(O',v),\frac{3}{4} g_1(O')\right\}\right)\right|\\
    \le&\mathbb{E}\left(\left|\min\left\{d(O,v),\frac{3}{4} g_1(O)\right\}-\min\left\{d(O',v),\frac{3}{4} g_1(O')\right\}\right|\right)\\
    \le &\mathbb{E}\left(\max\left\{\left|d(O,v)-d(O',v)\right|,\frac{3}{4} \left|g_1(O)- g_1(O')\right|\right\}\right)\\
    \le &\mathbb{E}\left(\max\left\{d(O,O'),\frac{3}{4}d(O,O')\right\}\right)\\
    =&d(O,O').
\end{align*}
\end{proof}
\begin{lemma}\label{lem:Lipschitz-g3}
For any $O,O'\in \mathbb{R}^2$, we have \[\left|g_3(O)-g_3(O')\right|\le (3+\sqrt{2})\;d(O,O').\]
\end{lemma}
\begin{proof}
For each $O\in\mathbb{R}$ and each $r\in \mathbb{R}^+$, let $C(O,r)$ denote the circle $\{x\in\mathbb{R}^2:d(O,x)=r\}$. Furthermore, let $B(O,r)$ denote the intersection of $[0,1]^2$ and the interior of $C(O,r)$, and $\gamma(O,r)$ denote the intersection of $[0,1]^2$ and $C(O,r)$. 

Sincec $B(O,r)$ is a convex subset of the unit square $[0,1]^2$, by an axiom\footnote{The axiom can be equivalently stated as follows: given two nested convex closed curves, the inner one is shorter. Archimedes applied this axiom to show that \cite{Archimedes2,richeson2015circular} the ratio of a circle's circumference to its diameter is a constant, namely Archimedes' constant $\pi$. He then obtained \cite{Archimedes2} the first rigorous approximation of $\pi$. Nowadays, there are many modern proofs of this axiom.} of Archimedes \cite{Archimedes}, the length of the boundary of $B(O,r)$ is at most $4$. Because $\gamma(O,r)$ is part of the boundary of $B(O,r)$, the length $\norm{\gamma(O,r)}$ is at most $4$. Since $\norm{\gamma(O,r)}$ is the derivative of the measure of $B(O,r)$ with respect to $r$, the difference between the measures of $B(O,r)$ and $B(O,r')$ is at most $4|r-r'|$. 

For each $O\in\mathbb{R}$ and each $r\in \mathbb{R}^+$, let $\gamma_i(O,r)$ ($i=1,2,3,4$) denote the intersection of an edge of $[0,1]^2$ (left, right, bottom, top) and $C(O,r)$. Then the derivative of the measure of $B(O,r)$ with respect to $x$-coordinate (resp.\ $y$-coordinate) of $O$ is $\norm{\gamma_1}-\norm{\gamma_2}$ (resp.\ $\norm{\gamma_3}-\norm{\gamma_3}$). As a result, the difference between the measures of $B(O,r)$ and $B(O',r)$ is at most $\sqrt{2}\; d(O,O')$. 

By definition, the measure of $B\left(O,\frac{3}{4}g_1(O)\right)$ is $1-g_3(O)$, and the measure of $B\left(O',\frac{3}{4}g_1(O')\right)$ is $1-g_3(O')$. Therefore, by \cref{lem:Lipschitz-g1}, we have
\begin{align*}
     & \left|g_3(O)-g_3(O')\right|\\
    \le &\sqrt{2} \;d(O,O')+3 \left|g_1(O)-g_1(O')\right|\\
    \le&(3+\sqrt{2})\;d(O,O').
\end{align*}
\end{proof}
\subsection{Proofs Using the $\eps$-Net}
\label{sec:epsilon}
In this subsection, we prove \cref{lemma-lambda} and \cref{lemma-mv-2}. It suffices to show 
\begin{equation}\label{eq:g2}
    g_2(O)\ge \frac{31}{48}g_1(O)
\end{equation}
and
\begin{equation}\label{eq:g3}
    g_3(O)\ge \frac{31}{48}
\end{equation}
for every $O\in\mathbb{R}^2$.

Let $d(O,[0,1]^2)$ denote the distance between $O$ and $[0,1]^2$. Then we have
\[g_1(O)=\mathbb{E}\left(d(O,v)\right)\le d(O,[0,1]^2)+\sqrt{2}.\] If $d(O,[0,1]^2)$ is at least $3\sqrt{2}$, then we have \[R=\frac{3}{4} g_1(O)\le\frac{3}{4}\left(d(O,[0,1]^2)+\sqrt{2}\right)\le d(O,[0,1]^2). \]
Thus by definition, we have $g_2(O) = \frac{3}{4} g_1(O)$ and $g_3(O) = 1$. Therefore, both \cref{eq:g2} and \cref{eq:g3} hold.

It remains to consider the case when $d(O,[0,1]^2)$ is less than $3\sqrt{2}$. Without loss of generality, we also assume that $O$ is located to the right of the vertical line $x=\frac{1}{2}$ and above the diagonal line $y=x$.

Let $N$ denote the point set \[N:=\left\{(a,b)\in\mathbb{R}^2: \left(\frac{a-0.5}{0.002}, \frac{b-0.5}{0.002}\right)\in\{0,1,\ldots, 2371\}^2, a\le b \right\}.\]
For each $O'\in N$, we have verified \[g_2(O')-\frac{31}{48}g_1(O')\ge 0.0025\] and \[g_3(O')-\frac{31}{48}\ge 0.0096.\]  
This is done using a rigorous computer-assisted proof.\footnote{Lengthy discussions on rigorous computer-assisted proofs can be found in many works in the literature, see, e.g., Zwick's discussion on the role of computer-assisted proofs in
mathematics and computer science~\cite{zwick2002computer}, the computer-assisted proofs for the four color theorem~\cite{apple1997_1,appel1977_2} and Kepler's conjecture~\cite{hales2005proof}.}
For each $O'\in N$, we compute the values of $g_1(O'), g_2(O'), g_3(O')$ using the closed-form formulas derived in \cref{closed-form-g1} and  \cref{closed-form-g2-g3}; see \cref{subsec:closed-form}.
Our computation\footnote{Our verification code is available at \url{https://github.com/Zipei-Nie/CVRP-proofs}. We use the kv library \cite{kv} for interval arithmetic.} carefully accounts for all the rounding errors using \emph{interval arithmetic}~\cite{hansen2003global}.

Consider a point $O$ in the set \[\left\{(a,b)\in \mathbb{R}^2:d((a,b),[0,1]^2)<3\sqrt{2},  \frac{1}{2}\le a\le b\right\}.\]
By the construction of $N$, there exists $O'\in N$ such that $d(O,O')\le \frac{\sqrt{2}}{1000}$.
By \cref{lem:Lipschitz-g1} and \cref{lem:Lipschitz-g2}, we have
\begin{align*}
    &g_2(O)-\frac{31}{48}g_1(O)\\
   =&g_2(O')-\frac{31}{48}g_1(O') +(g_2(O)-g_2(O'))-\frac{31}{48}(g_1(O)-g_1(O'))\\
   \ge& 0.0025 -\frac{79}{48} d(O,O')\\
   \ge& 0.0025-\frac{79\sqrt{2}}{48000}\\
   \ge& 0.
\end{align*}
By \cref{lem:Lipschitz-g3}, we have
\begin{align*}
    &g_3(O)-\frac{31}{48}\\
   =&g_3(O')-\frac{31}{48} +(g_3(O)-g_3(O'))\\
   \ge& 0.0096 -(3+\sqrt{2}) d(O,O')\\
   \ge& 0.0096-\frac{(3+\sqrt{2})\sqrt{2}}{1000}\\
   \ge& 0.
\end{align*}
Thus \cref{eq:g2} and \cref{eq:g3} hold. 

We completed the proofs of \cref{lemma-lambda} and \cref{lemma-mv-2}.

\subsection{Closed-Form Formulas for $g_1$, $g_2$, and $g_3$}\label{subsec:closed-form}
In this subsection, we will derive closed-form formulas for the functions $g_1$, $g_2$, and $g_3$. 

First, we compute the integrations of $1$ and $\sqrt{x^2+y^2}$ over a right triangle. 

\begin{definition}
Define the functions $A_0,A_1:\mathbb{R}^2\to \mathbb{R}$ by
\[A_0(a,b):=\begin{cases}\int_{0}^a \int_{0}^{\frac{b x}{a}} \,dy \,dx&\mbox{, if }a\neq 0,\\
0&\mbox{, if } a=0
\end{cases}\]
and
\[A_1(a,b):=\begin{cases}\int_{0}^a \int_{0}^{\frac{b x}{a}} \sqrt{x^2+y^2} \,dy \,dx&\mbox{, if }a\neq 0,\\
0&\mbox{, if } a=0
\end{cases}\]
for every $a,b\in \mathbb{R}$.
\end{definition}

Clearly, we have \[A_0(a,b)= \frac{ab}{2}.\] To obtain a closed-form formula for $A_1(a,b)$, we can use the result derived by Stone \cite{stone1991some}. See \cite{li2016moments} for an alternative proof.

\begin{lemma}\label{lem:right-triangle}\cite{stone1991some}
For every $a,b\in \mathbb{R}$ with $a\neq 0$, we have
\[
    A_1(a,b)=\begin{cases}\frac{a^3}{6} \log(\frac{b}{|a|}+\sqrt{1+\frac{b^2}{a^2}})+ \frac{ab}{6}\sqrt{a^2+b^2}&\mbox{, if }a\neq 0,\\
0&\mbox{, if } a=0.
\end{cases}\]
\end{lemma}

Then we have the following closed-form formula for $g_1$. See \cite{stone1991some} for an alternative formulation of the result.

\begin{theorem}\label{closed-form-g1}
For any $O=(a,b)\in \mathbb{R}^2$, we have
\begin{align*}
    g_1(O)=& A_1(a,b)+A_1(b,a)+A_1(b,1-a)+A_1(1-a,b)\\&+A_1(1-a,1-b)+A_1(1-b,1-a)+ A_1(1-b,a)+A_1(a,1-b).
\end{align*}
\end{theorem}
\begin{proof}
Divide the square $[0,1]^2$ into eight right triangles. Each triangle has one vertex at $O$ and one vertex at a corner of the square, and has one edge parallel to the $x$-axis and one edge parallel to the $y$-axis. 

By the definition of $g_1(O)$, we have \[g_1(O)=\int_0^1\int_0^1 \sqrt{(x-a)^2+(y-b)^2}\, dx\, dy.\]
We can break down this integration into eight separate ones, one for each of the right triangles formed by dividing the square. Since each integration has the form $A_1(\cdot, \cdot)$, the closed-form formula for $g_1$ follows from Lemma \ref{lem:right-triangle}.
\end{proof}

Then, we derive closed-form formulas for the integrations of $1$ and $\sqrt{x^2+y^2}$ over a disk segment. 

\begin{definition}
For each $h\in \mathbb{R}$, let $S_h$ denote the disk segment \[S_h:= \{(x,y)\in \mathbb{R}^2 : x\le h, x^2+y^2\le 1 \}.\]
We define the functions $B_0,B_1:\mathbb{R}\to \mathbb{R}$ by
\[B_0(h):=\iint_{S_h} dx\,dy\]
and 
\[B_1(h):=\iint_{S_h} \sqrt{x^2+y^2} \,dx\,dy\]
for every $h\in\mathbb{R}$.
\end{definition}

\begin{lemma}\label{closed-form-B}
For each $i=0,1$ and each $h\in \mathbb{R}$, we have \[B_i(h)=\begin{cases}0 &\mbox{, if } h< -1,\\ \frac{3-i}{3}\left(\pi -\arccos h\right)+ 2 A_i\left(h,\sqrt{1-h^2}\right) &\mbox{, if } -1\le h < 1,\\ \frac{3-i}{3}\pi &\mbox{, if } h \ge 1.\end{cases}\]
\end{lemma}
\begin{proof}
For $h < -1$, the region $S_h$ is empty, so $B_0(h)=B_1(h)=0$. For $h\ge -1$, the region $S_h$ is the unit disk, so $B_0(h)=\pi$ and $B_1(h)=\frac{2\pi}{3}$.

Suppose that $-1\le h < 1$. Divide $S_h$ into a disk sector and two right triangles. Each triangle has one vertex at the origin and one vertex at an intersection point of the unit circle and the line $x=h$, and has one edge on the $x$-axis and one edge on the line $x=h$. 

We can break down the integrations over $S_h$ into three separate ones, one for the disk sector and one for each right triangle. The disk sector has central angle $2(\pi-\arccos h)$ and radius $1$, so the area is $\pi-\arccos h$. The average distance from a uniform random point on the disk sector to the center is $\frac{2}{3}$ times the radius, so the integration of $\sqrt{x^2+y^2}$ over the disk sector is $\frac{2}{3}(\pi-\arccos h)$. The integrations over the right triangles have the form $A_i(\cdot,\cdot)$ ($i=0,1$), so we can use Lemma \ref{lem:right-triangle} to derive the closed-form formulas for $B_0$ and $B_1$.
\end{proof}

Next, we derive closed-form formulas for the integrations of $1$ and $\sqrt{x^2+y^2}$ over a region formed by the intersection of a disk and two half-planes. 

\begin{definition}
For each $h_1,h_2\in \mathbb{R}$, let $S_{h_1,h_2}$ denote the region \[S_{h_1,h_2}:= \{(x,y)\in \mathbb{R}^2 : x\le h_1,y\le h_2, x^2+y^2\le 1 \}.\]
We define the functions $C_0,C_1:\mathbb{R}^2\to \mathbb{R}$ by
\[C_0(h_1,h_2):=\iint_{S_{h_1,h_2}} dx\,dy\]
and 
\[C_1(h_1,h_2):=\iint_{S_{h_1,h_2}} \sqrt{x^2+y^2} \,dx\,dy\]
for every $h_1,h_2\in\mathbb{R}$.
\end{definition}
\begin{lemma}\label{closed-form-C}
For each $i=0,1$ and each $h_1,h_2\in \mathbb{R}$, we have
\[C_i(h_1,h_2)=
\begin{cases}
 0 &\mbox{, if }h_1^2+h_2^2 > 1, h_1\le 0, h_2\le 0,\\
 B_i(h_2) &\mbox{, if }h_1^2+h_2^2 > 1, h_1>0, h_2\le 0,\\
 B_i(h_1)&\mbox{, if }h_1^2+h_2^2 > 1, h_1\le 0, h_2>0,\\
B_i(h_1)+B_i(h_2)-\frac{3-i}{3}\pi&\mbox{, if } h_1^2+h_2^2 > 1 ,h_1>0,h_2>0,\\
    {\frac{3-i}{6}\left(\frac{\pi}{2}+\arcsin h_1+\arcsin h_2\right) + A_i\left(h_1,\sqrt{1-h_1^2}\right)} \\{+ A_i\left(h_2,\sqrt{1-h_2^2}\right)+A_i(h_1,h_2)+A_i(h_2,h_1)}
 &\mbox{, if } h_1^2+h_2^2\le 1.\\
\end{cases}\]
\end{lemma}
\begin{proof}
The proof is very similar to those of \cref{closed-form-g1} and \cref{closed-form-B}.

For $h_1^2+h_2^2> 1, h_1 \le 0, h_2\le 0$, the region $S_{h_1,h_2}$ is empty. For $h_1^2+h_2^2> 1, h_1 > 0, h_2\le 0$, the region $S_{h_1,h_2}$ is a disk segment. The situation is symmetric when $h_1^2+h_2^2> 1, h_1 \le 0, h_2> 0$.

For $h_1^2+h_2^2>1, h_1>0, h_2>0$, we divide the region $S_{h_1,h_2}$ into two disk segments and a negative disk. For $h_1^2+h_2^2\le 1$, we divide the region $S_{h_1,h_2}$ into a disk sector and four right triangles.

In any case, we can use \cref{lem:right-triangle} and \cref{closed-form-B} to derive the closed-form formulas for $C_0$ and $C_1$.
\end{proof}

Then, we obtain closed-form formulas for integrations directly related to the definitions of $g_2$ and $g_3$.

\begin{definition}
Define the functions $D_0,D_1:\mathbb{R}^2\times \mathbb{R}^+\to \mathbb{R}$ by
\[D_0(a,b,R)=\frac{1}{R^2}\int_{0}^1 \int_0^1 \mathds{1}\left(d(O,v)\le R\right)\,dx\,dy\]
and 
\[D_1(a,b,R)=\frac{1}{R^3}\int_{0}^1 \int_0^1 d(O,v)\mathds{1}\left(d(O,v)\le R\right)\,dx\,dy\]
for every $a,b\in \mathbb{R}$ and $R\in\mathbb{R}^+$.
\end{definition}

\begin{lemma}\label{closed-form-D}
    For each $i=0,1$, each $a,b\in \mathbb{R}$, and each $R\in\mathbb{R}^+$, we have
    \[D_i(a,b,R)=C_i\left(\frac{1-a}{R},\frac{1-b}{R}\right)-C_i\left(\frac{1-a}{R},\frac{-b}{R}\right)-C_i\left(\frac{-a}{R},\frac{1-b}{R}\right)+C_i\left(\frac{-a}{R},\frac{-b}{R}\right).\]
\end{lemma}
\begin{proof}
    We divide the unit square $[0,1]^2$ into two positive regions, $(-\infty, 0)\times (-\infty,0)$ and $(-\infty, 1]\times (-\infty,1]$, and two negative regions, $(-\infty, 1]\times (-\infty,0)$ and $(-\infty, 0)\times (-\infty,1]$. Over the positive regions, the integrations are $C_i\left(\frac{1-a}{R},\frac{1-b}{R}\right)$ and $+C_i\left(\frac{-a}{R},\frac{-b}{R}\right)$ respectively. Over the negative regions, the integrations are $-C_i\left(\frac{1-a}{R},\frac{-b}{R}\right)$ and $-C_i\left(\frac{-a}{R},\frac{1-b}{R}\right)$ respectively. We can use \cref{closed-form-C} to derive the closed-form formulas for $D_0$ and $D_1$.
\end{proof} 

Finally, we establish the closed-form formulas for $g_2$ and $g_3$.
\begin{theorem}\label{closed-form-g2-g3}
For any $O=(a,b)\in \mathbb{R}^2$, we have  
    \[g_2(O)=R-R^3 D_0(a,b,R)+R^3 D_1(a,b,R)\]
and
    \[g_3(O)=1-R^2 D_0(a,b,R),\]
where \[R=\frac{3}{4}g_1(O).\]
\end{theorem}
\begin{proof}
By definition, we have
\[g_2(O)=\int_{0}^1 \int_0^1 \min\left\{d(O,v),R\right\}\,dx\,dy\]
and 
\[g_3(O)=\int_{0}^1 \int_0^1 \mathds{1}\left(d(O,v)>R\right)\,dx\,dy.\]
The closed-form formulas for $g_2(O)$ and $g_3(O)$ follows from \cref{closed-form-g1}, \cref{closed-form-D}, and the equations
\[\min\left\{d(O,v),R\right\}= R\;\mathds{1}\left(d(O,v)> R\right)+d(O,v) \mathds{1}\left(d(O,v)\le R\right)\]
and 
\[\mathds{1}\left(d(O,v)> R\right)=1-\mathds{1}\left(d(O,v)\le R\right).\]
\end{proof}

\subsection*{Acknowledgments}
We thank Claire Mathieu for helpful preliminary discussions.

\end{appendix}

\bibliographystyle{abbrvurl}
\bibliography{references}

\begin{thebibliography}{10}

\bibitem{adamaszek2010ptas}
A.~Adamaszek, A.~Czumaj, and A.~Lingas.
\newblock {PTAS} for $k$-tour cover problem on the plane for moderately large
  values of $k$.
\newblock {\em International Journal of Foundations of Computer Science},
  21(06):893--904, 2010.

\bibitem{altinkemer1987heuristics}
K.~Altinkemer and B.~Gavish.
\newblock Heuristics for unequal weight delivery problems with a fixed error
  guarantee.
\newblock {\em Operations Research Letters}, 6(4):149--158, 1987.

\bibitem{altinkemer1990heuristics}
K.~Altinkemer and B.~Gavish.
\newblock Heuristics for delivery problems with constant error guarantees.
\newblock {\em Transportation Science}, 24(4):294--297, 1990.

\bibitem{anbuudayasankar2016models}
S.~P. Anbuudayasankar, K.~Ganesh, and S.~Mohapatra.
\newblock {\em Models for practical routing problems in logistics}.
\newblock Springer, 2016.

\bibitem{apple1997_1}
K.~Appel and W.~Haken.
\newblock {Every planar map is four colorable. Part I: Discharging}.
\newblock {\em Illinois Journal of Mathematics}, 21(3):429 -- 490, 1977.

\bibitem{appel1977_2}
K.~Appel, W.~Haken, and J.~Koch.
\newblock {Every planar map is four colorable. Part II: Reducibility}.
\newblock {\em Illinois Journal of Mathematics}, 21(3):491--567, 1977.

\bibitem{Archimedes}
Archimedes.
\newblock {\em On the sphere and cylinder, Book I}.
\newblock c. 225 BCE.

\bibitem{Archimedes2}
Archimedes.
\newblock {\em Measurement of a Circle}.
\newblock c. 250 BCE.

\bibitem{arora1998polynomial}
S.~Arora.
\newblock Polynomial time approximation schemes for {Euclidean} traveling
  salesman and other geometric problems.
\newblock {\em Journal of the ACM (JACM)}, 45(5):753--782, 1998.

\bibitem{arora2003approximation}
S.~Arora.
\newblock Approximation schemes for {NP}-hard geometric optimization problems:
  A survey.
\newblock {\em Mathematical Programming}, 97(1):43--69, 2003.

\bibitem{asano2001new}
T.~Asano, N.~Katoh, and K.~Kawashima.
\newblock A new approximation algorithm for the capacitated vehicle routing
  problem on a tree.
\newblock {\em Journal of Combinatorial Optimization}, 5(2):213--231, 2001.

\bibitem{asano1997covering}
T.~Asano, N.~Katoh, H.~Tamaki, and T.~Tokuyama.
\newblock Covering points in the plane by $k$-tours: towards a polynomial time
  approximation scheme for general $k$.
\newblock In {\em ACM Symposium on Theory of Computing (STOC)}, pages 275--283,
  1997.

\bibitem{beardwood1959shortest}
J.~Beardwood, J.~H. Halton, and J.~M. Hammersley.
\newblock The shortest path through many points.
\newblock In {\em Mathematical Proceedings of the Cambridge Philosophical
  Society}, volume~55, pages 299--327. Cambridge University Press, 1959.

\bibitem{becker2018tight}
A.~Becker.
\newblock {A tight 4/3 approximation for capacitated vehicle routing in trees}.
\newblock In {\em International Conference on Approximation Algorithms for
  Combinatorial Optimization Problems}, volume 116, pages 3:1--3:15, 2018.

\bibitem{becker2017quasi}
A.~Becker, P.~N. Klein, and D.~Saulpic.
\newblock A quasi-polynomial-time approximation scheme for vehicle routing on
  planar and bounded-genus graphs.
\newblock In {\em 25th Annual European Symposium on Algorithms (ESA)}, 2017.

\bibitem{becker2018polynomial}
A.~Becker, P.~N. Klein, and D.~Saulpic.
\newblock Polynomial-time approximation schemes for $k$-center, $k$-median, and
  capacitated vehicle routing in bounded highway dimension.
\newblock In {\em 26th Annual European Symposium on Algorithms (ESA)}, 2018.

\bibitem{becker2019ptas}
A.~Becker, P.~N. Klein, and A.~Schild.
\newblock {A PTAS for bounded-capacity vehicle routing in planar graphs}.
\newblock In {\em Workshop on Algorithms and Data Structures}, pages 99--111.
  Springer, 2019.

\bibitem{becker2019framework}
A.~Becker and A.~Paul.
\newblock A framework for vehicle routing approximation schemes in trees.
\newblock In {\em Workshop on Algorithms and Data Structures}, pages 112--125.
  Springer, 2019.

\bibitem{blauth2023improving}
J.~Blauth, V.~Traub, and J.~Vygen.
\newblock Improving the approximation ratio for capacitated vehicle routing.
\newblock {\em Mathematical Programming}, 197(2):451--497, 2023.

\bibitem{bompadre2006improved}
A.~Bompadre, M.~Dror, and J.~B. Orlin.
\newblock Improved bounds for vehicle routing solutions.
\newblock {\em Discrete Optimization}, 3(4):299--316, 2006.

\bibitem{bompadre2007probabilistic}
A.~Bompadre, M.~Dror, and J.~B. Orlin.
\newblock Probabilistic analysis of unit-demand vehicle routeing problems.
\newblock {\em Journal of applied probability}, 44(1):259--278, 2007.

\bibitem{cohen2020light}
V.~{Cohen-Addad}, A.~Filtser, P.~N. Klein, and H.~Le.
\newblock On light spanners, low-treewidth embeddings and efficient traversing
  in minor-free graphs.
\newblock In {\em Symposium on Foundations of Computer Science (FOCS)}, pages
  589--600. IEEE, 2020.

\bibitem{cordeau2007vehicle}
J.-F. Cordeau, G.~Laporte, M.~W. Savelsbergh, and D.~Vigo.
\newblock Vehicle routing.
\newblock {\em Handbooks in operations research and management science},
  14:367--428, 2007.

\bibitem{crainic2012fleet}
T.~G. Crainic and G.~Laporte.
\newblock {\em Fleet management and logistics}.
\newblock Springer Science \& Business Media, 2012.

\bibitem{daganzo1984distance}
C.~F. Daganzo.
\newblock The distance traveled to visit n points with a maximum of c stops per
  vehicle: An analytic model and an application.
\newblock {\em Transportation science}, 18(4):331--350, 1984.

\bibitem{das2015quasipolynomial}
A.~Das and C.~Mathieu.
\newblock A quasi-polynomial time approximation scheme for {Euclidean}
  capacitated vehicle routing.
\newblock {\em Algorithmica}, 73(1):115--142, 2015.

\bibitem{DMZ23}
M.~Dufay, C.~Mathieu, and H.~Zhou.
\newblock An approximation algorithm for distance-constrained vehicle routing
  on trees.
\newblock In {\em 40th International Symposium on Theoretical Aspects of
  Computer Science (STACS 2023)}, volume 254, pages 27:1--27:16, 2023.

\bibitem{friggstad2021improved}
Z.~Friggstad, R.~Mousavi, M.~Rahgoshay, and M.~R. Salavatipour.
\newblock Improved approximations for capacitated vehicle routing with
  unsplittable client demands.
\newblock In {\em International Conference on Integer Programming and
  Combinatorial Optimization}, pages 251--261. Springer, 2022.

\bibitem{friggstad2014approximation}
Z.~Friggstad and C.~Swamy.
\newblock Approximation algorithms for regret-bounded vehicle routing and
  applications to distance-constrained vehicle routing.
\newblock In {\em Proceedings of the forty-sixth annual ACM Symposium on Theory
  of Computing (STOC)}, pages 744--753, 2014.

\bibitem{gillett1974heuristic}
B.~E. Gillett and L.~R. Miller.
\newblock A heuristic algorithm for the vehicle-dispatch problem.
\newblock {\em Operations Research}, 22(2):340--349, 1974.

\bibitem{golden2008vehicle}
B.~Golden, S.~Raghavan, and E.~Wasil.
\newblock {\em The vehicle routing problem: latest advances and new
  challenges}, volume~43 of {\em Operations Research/Computer Science
  Interfaces Series}.
\newblock Springer, 2008.

\bibitem{GMZ2023}
F.~Grandoni, C.~Mathieu, and H.~Zhou.
\newblock Unsplittable {Euclidean} capacitated vehicle routing: A
  $(2+\varepsilon)$-approximation algorithm.
\newblock In {\em 14th Innovations in Theoretical Computer Science Conference
  (ITCS 2023)}, volume 251, pages 63:1--63:13, 2023.

\bibitem{haimovich1985bounds}
M.~Haimovich and A.~H.~G. Rinnooy~Kan.
\newblock Bounds and heuristics for capacitated routing problems.
\newblock {\em Mathematics of operations Research}, 10(4):527--542, 1985.

\bibitem{hales2005proof}
T.~C. Hales.
\newblock {A proof of the Kepler conjecture}.
\newblock {\em Annals of mathematics}, pages 1065--1185, 2005.

\bibitem{hansen2003global}
E.~Hansen and G.~W. Walster.
\newblock {\em Global optimization using interval analysis: revised and
  expanded}, volume 264.
\newblock CRC Press, 2003.

\bibitem{jayaprakash2023approximation}
A.~Jayaprakash and M.~R. Salavatipour.
\newblock Approximation schemes for capacitated vehicle routing on graphs of
  bounded treewidth, bounded doubling, or highway dimension.
\newblock {\em ACM Trans. Algorithms}, 19(2), 2023.

\bibitem{karp1977probabilistic}
R.~M. Karp.
\newblock Probabilistic analysis of partitioning algorithms for the
  traveling-salesman problem in the plane.
\newblock {\em Mathematics of operations research}, 2(3):209--224, 1977.

\bibitem{kv}
M.~Kashiwagi.
\newblock {kv -- a C++ Library for Verified Numerical Computation}.
\newblock \url{http://verifiedby.me/kv/index-e.html}.
\newblock Accessed April 2023.

\bibitem{khachay2016ptas}
M.~Khachay and R.~Dubinin.
\newblock {PTAS for the Euclidean} capacitated vehicle routing problem in
  $\mathbb{R}^d$.
\newblock In {\em International Conference on Discrete Optimization and
  Operations Research}, pages 193--205. Springer, 2016.

\bibitem{laporte1992vehicle}
G.~Laporte.
\newblock The vehicle routing problem: An overview of exact and approximate
  algorithms.
\newblock {\em European journal of operational research}, 59(3):345--358, 1992.

\bibitem{laporte2000classical}
G.~Laporte, M.~Gendreau, J.-Y. Potvin, and F.~Semet.
\newblock Classical and modern heuristics for the vehicle routing problem.
\newblock {\em International transactions in operational research},
  7(4-5):285--300, 2000.

\bibitem{li1990worst}
C.~L. Li and D.~Simchi-Levi.
\newblock Worst-case analysis of heuristics for multidepot capacitated vehicle
  routing problems.
\newblock {\em ORSA Journal on Computing}, 2(1):64--73, 1990.

\bibitem{li1992distance}
C.-L. Li, D.~{Simchi-Levi}, and M.~Desrochers.
\newblock On the distance constrained vehicle routing problem.
\newblock {\em Operations Research}, 40(4):790--799, 1992.

\bibitem{li2016moments}
H.~Li and X.~Qiu.
\newblock Moments of distance from a vertex to a uniformly distributed random
  point within arbitrary triangles.
\newblock {\em Mathematical Problems in Engineering}, 2016, 2016.

\bibitem{MZ22}
C.~Mathieu and H.~Zhou.
\newblock {Iterated tour partitioning for Euclidean capacitated vehicle
  routing}.
\newblock {\em Random Structures \& Algorithms}, pages 1--20, 2022.
\newblock \href {https://doi.org/10.1002/rsa.21130}
  {\path{doi:10.1002/rsa.21130}}.

\bibitem{MZ23_PTAS}
C.~Mathieu and H.~Zhou.
\newblock A {PTAS} for capacitated vehicle routing on trees.
\newblock {\em ACM Transactions on Algorithms (TALG)}, 19(2), 2023.

\bibitem{MZ22_unsplittable}
C.~Mathieu and H.~Zhou.
\newblock A tight $(1.5 + \varepsilon)$-approximation for unsplittable
  capacitated vehicle routing on trees.
\newblock {\em 50th International Colloquium on Automata, Languages, and
  Programming (ICALP 2023)}, 2023.

\bibitem{MZ23graphic}
T.~M{\"o}mke and H.~Zhou.
\newblock Capacitated vehicle routing in graphic metrics.
\newblock In {\em Symposium on Simplicity in Algorithms (SOSA)}, pages
  114--123. SIAM, 2023.

\bibitem{nagarajan2012approximation}
V.~Nagarajan and R.~Ravi.
\newblock Approximation algorithms for distance constrained vehicle routing
  problems.
\newblock {\em Networks}, 59(2):209--214, 2012.

\bibitem{viger}
{Private communication with Fabien Viger at Google}, 2022.

\bibitem{rhee1994probabilistic}
W.~T. Rhee.
\newblock Probabilistic analysis of a capacitated vehicle routing problem {II}.
\newblock {\em The Annals of Applied Probability}, 4(3):741--764, 1994.

\bibitem{richeson2015circular}
D.~Richeson.
\newblock Circular reasoning: who first proved that {C} divided by d is a
  constant?
\newblock {\em The College Mathematics Journal}, 46(3):162--171, 2015.

\bibitem{schneider2014convex}
R.~Schneider.
\newblock {\em Convex bodies: the {Brunn--Minkowski} theory}.
\newblock Number 151. Cambridge university press, 2014.

\bibitem{stone1991some}
R.~E. Stone.
\newblock Some average distance results.
\newblock {\em Transportation Science}, 25(1):83--90, 1991.

\bibitem{szasz1917extremaleigenschaft}
O.~Szasz and A.~Rosenthal.
\newblock Eine extremaleigenschaft der kurven konstanter breite.
\newblock {\em Jahresbericht der Deutschen Mathematiker-Vereinigung},
  25:278--282, 1917.

\bibitem{toth2002vehicle}
P.~Toth and D.~Vigo.
\newblock {\em The Vehicle Routing Problem}.
\newblock Society for Industrial and Applied Mathematics, 2002.

\bibitem{wren1971computers}
A.~Wren.
\newblock {\em Computers in Transport Planning and Operation}.
\newblock Allan, 1971.

\bibitem{wren1972computer}
A.~Wren and A.~Holliday.
\newblock Computer scheduling of vehicles from one or more depots to a number
  of delivery points.
\newblock {\em Journal of the Operational Research Society}, 23(3):333--344,
  1972.

\bibitem{zwick2002computer}
U.~Zwick.
\newblock Computer assisted proof of optimal approximability results.
\newblock In {\em SODA}, pages 496--505, 2002.

\end{thebibliography}
\end{document}